\documentclass[12pt]{l4dc2022}

\usepackage{graphics}
\usepackage{algorithm}
\usepackage{algpseudocode}
\usepackage{multirow}
\usepackage{mathrsfs}  
\usepackage{xcolor}
\usepackage{amsfonts}
\usepackage{balance}
\usepackage{cancel}
\usepackage{bm}
\usepackage[font=small,skip=0pt]{caption}
\usepackage{times}

\title[Gradient and Projection Free Distributed Online Min-Max Resource Optimization]{Gradient and Projection Free Distributed Online \\Min-Max Resource Optimization}

\author{%
 \Name{Jingrong Wang} \Email{jr.wang@mail.utoronto.ca}\\
 \Name{Ben Liang} \Email{liang@ece.utoronto.ca}\\
 \addr Department of Electrical and Computer Engineering, University of Toronto, Canada%
}

\begin{document}

\maketitle

\thispagestyle{plain}

\begin{abstract}%
	We consider distributed online min-max resource allocation with a set of parallel agents and a parameter server.
	Our goal is to minimize the pointwise maximum over a set of time-varying and decreasing cost functions, without a priori information about these functions.
	We propose a novel online algorithm, termed Distributed Online resource Re-Allocation (\texttt{DORA}), where non-stragglers learn to relinquish resource and share resource with stragglers.
	A notable feature of \texttt{DORA} is that it does not require gradient calculation or projection operation, unlike most existing online optimization strategies. 
	This allows it to substantially reduce the computation overhead in large-scale and distributed networks.
	 We analyze the worst-case performance of \texttt{DORA} and derive an upper bound on its dynamic regret for non-convex functions. 
	We further consider an application to the bandwidth allocation problem in distributed online machine learning.
	Our numerical study demonstrates the efficacy of the proposed solution and its performance advantage over gradient- and/or projection-based resource allocation algorithms in reducing wall-clock time.%
\end{abstract}


\section{Introduction}\label{sec: intro}

We consider an online min-max resource allocation problem with a set of agents.
The objective is to minimize the accumulation of the pointwise maximum over a set of \textit{time-varying} local cost functions of the agents, subject to a limited resource budget.
The cost functions are decreasing with respect to the resource allocated to that agent.
Optimization decisions are made round by round, while the cost functions vary arbitrarily over time and are revealed only \textit{after} decision making in each round. 
Our objective is to decide a sequence of allocation schemes as delayed information about the cost functions becomes available over time. 

Many practical applications provide strong motivation for this special type of minimization problem, e.g., minimum spanning tree~(\cite{Y98}), online routing~(\cite{H19}), fair resource allocation in wireless cellular networks~(\cite{HM05}), and data/task parallelism in distributed learning/computing~(\cite{DG08}). 
Take synchronous distributed learning in a parameter-server architecture as an example.
The learning model is updated only after all agents send their local gradients to the parameter server. 
Therefore, the training time in each round, which consists of both the computation and communication times, is determined by the \textit{straggler}, i.e., the last agent to send its local gradient.
Thus, communication resource allocation for distributed computing can be formulated as a min-max optimization problem, minimizing the worst delay among all agents~(\cite{DZ17}).
Furthermore, both the computation and communication times are time-varying in an unpredictable fashion.
This is mainly due to the wireless fading channels and the randomness of computing demands and available capacity.
In this dynamic environment, an \textit{online} min-max resource allocation algorithm is of interest.

Online min-max optimization belongs to the family of online optimization, which has been extensively studied in a wide variety of real-life applications, but most of the existing works focus on convex cost functions~(\cite{SS11}).
Typically, an online algorithm proceeds in a round-by-round manner.
In each round, the agent first selects a decision. 
Then, the environment reveals to the agent the cost function for this round, and the agent updates its strategy for the next round accordingly. 
In this context, \textit{regret} is introduced to measure the performance of an online algorithm versus a comparator.
In particular, the \textit{static} regret uses a fixed minimizer of the accumulated cost over time as the comparator.
However, considering the inherent dynamics of real-world systems, the \textit{dynamic} regret is practically more meaningful~(\cite{Z03}).
It measures the difference between the sequence generated by an online algorithm and a sequence of time-varying comparators, most commonly the instantaneous minimizers.

Online min-max optimization problems have unique characteristics.
First, the pointwise maximum among the set of cost functions is generally not continuously differentiable, so that we require subgradients instead of gradients.
Second, the performance guarantees of computationally efficient projection-free algorithms, such as (\cite{HK12}), require strongly convex and smooth cost functions, which do not apply to the min-max problem.
The need for \textit{distributed} implementation brings additional challenges.
Unlike in the more common min-sum optimization, the max operation couples multiple local cost functions, requiring some decomposition before problem-solving in a distributed manner. 
This is exacerbated by the coupling constraints, e.g., due to shared communication bandwidth among the agents, which adds to the unpredictability of the time-varying global cost function in the online setting.

In response to the above challenges, we propose a new distributed online algorithm named Distributed Online resource Re-Allocation (\texttt{DORA}), to solve an online min-max optimization problem in a multi-agent system with coupling linear constraints. 
Our main contributions include the following:
\textit{First}, \texttt{DORA} is a distributed algorithm that leverages the unique structure of the online min-max  optimization problem.
We denote the agent with the highest cost as \textit{straggler}.
Under \texttt{DORA}, the agents relinquish a carefully chosen amount of resource to the straggler in the previous online round to help improve its performance. 
To achieve this, each agent only obtains an incomplete view of the system and does not need to keep a full copy of the others' decisions.
More importantly, \texttt{DORA} does not require gradient calculation or projection operation, so it has much lower computational complexity than generic online solutions.
\textit{Second},  We analyze the dynamic regret of \texttt{DORA} for general decreasing cost functions. 
This compares favorably with the best existing solutions considering \texttt{DORA}'s reduced computational complexity, making it an attractive alternative.
\textit{Finally},  We apply the proposed solution to online resource allocation in distributed machine learning, which trains a convolutional neural network model for image identification with a set of edge devices in a parameter-server architecture. 
Our experimental results demonstrate the performance advantage of \texttt{DORA} over existing online algorithms that are based on gradient and/or projection, in terms of significantly faster convergence and reduced training time.

\section{Background and Related Work}\label{sec: related}

\subsection{Dynamic Regret of Gradient-based and Projection-based Online Optimization}

In online optimization, an agent aims to make a sequence of decisions to minimize the accumulation of a sequence of cost functions, without knowing the cost functions ahead of time. 
In each round $t$, the agent selects a decision $\mathbf{x}_t$ from a feasible set $\mathcal{F}$.
Then, the environment reveals cost function $f_t$ to the agent, and the agent suffers an instantaneous loss $f_t(\mathbf{x}_t)$ in this round. 

Regret is introduced to measure the performance of an online algorithm.
We focus on the more practically meaningful dynamic regret.
In recent literature, the dynamic regret is commonly defined as
$\textnormal{Reg}^d_T=\sum_{t\in\mathcal{T}}f_t(\mathbf{x}_t)-\sum_{t\in\mathcal{T}}f_t(\mathbf{x}_t^*)$,
where $\mathbf{x}_t^*\in \arg\min_{\mathbf{x}\in\mathcal{F}} f_t(\mathbf{x})$. 
Since the cost functions fluctuate arbitrarily, the dynamic regret in the worst case scales linearly. 
Typically, it is related to some \textit{regularity measures}, which indicate how dynamic the environment is. 
For example, the path-length of the dynamic minimizers is defined as
	$P_T=\sum_{t=2}^{T}||\mathbf{x}_{t-1}^*-\mathbf{x}_{t}^*||_2$.
For centralized online convex optimization (OCO), \cite{Z03} proved that with online (projected) gradient descent (OGD), the dynamic regret is upper bounded by $O(\sqrt T  (1+P_T))$. 
\cite{ZL18} improved the bound to $O(\sqrt{T (1+P_T)})$ by running multiple OGD operations in parallel.
For centralized online non-convex optimization (ONCO), \cite{HM2020} proved that the dual averaging algorithm enjoys a dynamic bound of $O(T^{\frac{2}{3}}V_T^{1/3})$, where $V_T$ is the time-accumulated variation of the cost functions.
If the cost functions are pseudo-convex, \cite{GL18} improved the bound to $O(\sqrt{T (1+P_T)})$.
Additional dynamic regret bounds have also been derived for centralized OCO algorithms, e.g., \cite{MS16, ZY17, BG15}.

In distributed implementation, most of the recent works have proposed methods that provide dynamic regret guarantee under various assumptions on the convexity and smoothness of the objective functions (\cite{SJ17, ZR19, DB19, LJ20, SK20, EL20, YL20, LY21}).
To the best of our knowledge, for gradient/projection-based distributed OCO, the tightest known dynamic regret bound for general convex cost functions is $O(\sqrt{T (1+P_T)})$~(\cite{SJ17}).
There is no dynamic regret bound developed for distributed ONCO.

In all of the above works, the local solution of each agent in each round is updated with the calculated gradient and needs to be projected back to the domain of interest to restore feasibility.
However, even the projection onto a simplex requires non-trivial computation time (\cite{LY09}).
In contrast, our work does not require gradient calculation or projection operation.

\subsection{Gradient-free and Projection-free Online Optimization}


Existing gradient-free algorithms replace exact gradient calculation with efficient gradient approximation (\cite{LC20}).
The FKM algorithm, named after its authors, uses spherical gradient estimators and then conducts gradient descent with projection (\cite{FL05, HMZ20}).
For OCO, with carefully designed diminishing step sizes, FKM and its variants can achieve dynamic regret that is upper bounded by $O(T^{\frac{4}{5}}V_T^{\frac{1}{5}})$ (\cite{BG15}). 
For ONCO, \cite{HM21} showed that the dual averaging algorithm with bandit feedback achieves a dynamic regret bound of $O(T^{\frac{n+2}{n+3}}V_T^{\frac{1}{n+3}})$, where $n\geq 1$ is the dimension of the decision variables.
This bound is worse than that of gradient-based algorithms due to the biased gradient estimation.

Among projection-free solutions, the Frank-Wolfe method, also known as the conditional gradient method, replaces projections with linear optimization over the feasible set~(\cite{FW56}).
In the centralized setup, \cite{HK12} proposed an online conditional gradient (OCG) algorithm, achieving $O(T^{\frac{3}{4}})$ static regret for convex functions.
If the functions are convex and smooth, with the Follow-the-Perturbed-Leader method, the static regret and dynamic regret of OCG were improved to $O(T^{\frac{2}{3}})$~(\cite{HM20}) and $O(T^{\frac{2}{3}}V_T^{1/3})$~(\cite{WX21}), respectively.
There is no dynamic regret bound developed for projection-free ONCO.

In distributed implementation, \cite{ZZ17} proposed a distributed variant of OCG, achieving the same static regret bound as the centralized one.
It is worth noting that the regret bounds of OCG are also worse than those of projection-based algorithms, suggesting performance penalty in exchange for reducing computation complexity.
Since the above works require strongly convex and smooth cost functions, their performance guarantees do not apply to the min-max resource allocation problem in our work.
Furthermore, no dynamic regret bound is known for distributed projection-free OCG.

Unlike existing gradient/projection-free algorithms, we propose to update the decisions online by directly calculating the amount of resource that the agents can relinquish and share with the stragglers.
We will show in Section~\ref{ch: performance} that this can substantially reduce the computation time in comparison with FKM and OCG.

\subsection{Min-Max Optimization}
The offline min-max optimization problem, i.e., with fixed cost functions known ahead of time, has been studied in~\cite{LL06, SN11, SN13, NF19, LL20, WP21}. 
These methods cannot be applied to our online problem.
For online min-max convex optimization, \cite{BC19} directly applied OGD to the min-max vertex cover problem, using a subgradient whenever the gradient does not exist.
However, they overlooked the special structure of the min-max problem.
As a result, there was no improvement in the static regret over general online algorithms. 
There is no existing work on distributed online algorithms specifically designed for the min-max problem.

\section{Problem Formulation}\label{sec: system}

Consider a distributed system with a set of parallel agents $\mathcal{N}=\{1,..., N\}$ and a parameter server.
The time is slotted by rounds $\mathcal{T}=\{1,...,T\}$.

In each round $t$, let $x_{i,t}$ denote the amount of resource allocated to agent $i$. 
Corresponding to $x_{i,t}$, the local cost function at agent $i$ is $f_{i,t}(x_{i,t})$. 
We note that $f_{i,t}$ may depend on other random factors, so it varies over time.
We assume that $f_{i,t}$ at any time $t$ is decreasing, but not necessarily strictly decreasing, in $x_{i,t}$. 
As shown in Section \ref{ch: performance}, an important example of decreasing cost function is where $x_{i,t}$ represents the allocated channel bandwidth and $f_{i,t}$ represents the communication delay plus any quantity independent of $x_{i,t}$.

Let $\mathbf{x}_t = (x_{1,t},...,x_{N,t})$. 
Then, the decision variables $\{\mathbf{x}_t\}_{t\in\mathcal{T}}$ is a sequence of resource allocation schemes.
The global cost function in round $t$ is the pointwise maximum over the set of time-varying local cost functions $\{f_{i,t},~ i\in\mathcal{N}\}$, i.e., 
\begin{equation}\label{eq: max}
	\begin{aligned}
		f_t(\mathbf{x}_t) = \max_{i\in\mathcal{N}}f_{i,t}(x_{i,t}).
	\end{aligned}
\end{equation}
As an example, in  edge learning, since each round is completed only when the last agent's task is finished, $f_t(\mathbf{x}_t)$ represents the duration of round $t$. 

We assume that the agents can observe $f_{i,t}$ but only at the end of round $t$ after each agent has been allocated $x_{i,t}$ and completed its task. 
In other words, we do not assume prior knowledge on how $x_{i,t}$ is mapped to $f_{i,t}$. 
This matches the standard form of online optimization and is a natural system model in many practical applications, e.g, due to the uncertain available processing power and wireless channel conditions in edge learning. 
Furthermore, once all agents complete their tasks at the end of round $t$, all local costs $f_{i,t}(x_{i,t})$, and thus the global cost in round $t$, are known to the parameter server.

Our objective is to minimize the accumulation of the sequence of global cost functions over time, subject to a limited resource budget. 
For notation simplicity, the resource budget is normalized to 1.
Then, the problem is formulated as follows:
\begin{align}
	\min_{\{\mathbf{x}_t\}_{t\in\mathcal{T}}} \quad & \sum_{t\in\mathcal{T}}f_t(\mathbf{x}_t),  \label{eq: minmax}\\
	\textrm{s.t.} \quad & \sum_{i\in\mathcal{N}}x_{i,t}\leq 1, \forall t\in\mathcal{T},~\label{eq:sum}\\
	& x_{i,t}\geq 0, \forall i \in \mathcal{N}, \forall t\in\mathcal{T}.~\label{eq:x_i}  
\end{align}

As explained in Section \ref{sec: intro}, there are broad applications to this optimization formulation.
If the information of local cost functions were known a priori, this problem could be solved as a traditional offline optimization problem.
However, since in many applications the cost functions become known only at the end of a round, we seek an \textit{online} solution to this problem.
Moreover, in large-scale systems, a centralized solution, e.g., one completely computed by the parameter server, would require intense information exchange that incurs high communication overhead and high computation requirement at the parameter server.
Instead, we are interested in developing a \textit{distributed} algorithm.
As explained in Section \ref{sec: related}, existing distributed online solutions are inefficient for this problem. Therefore, we will next present a new gradient-free and projection-free algorithm that substantially reduces the required computation.

\section{Distributed Online Resource Re-Allocation}\label{sec: algo}
In this section, we present the design of \texttt{DORA} to solve the online min-max optimization problem in (\ref{eq: minmax})-(\ref{eq:x_i}).
In \texttt{DORA}, the agents simultaneously learn to track the unknown local cost functions over time and solve the problem with help from the parameter server.

\subsection{Reformulation of the Cost Function in Round \textit{t}}\label{sec: decompose}

To transform the global cost function toward distributed computation, we first consider the \textit{instantaneous} optimization problem in each round $t$, i.e., where the objective is $\min_{\mathbf{x}_t} \max_{i\in\mathcal{N}} f_{i,t}(x_{i,t})$.
The equivalent epigraph representation of this problem is
\begingroup
\allowdisplaybreaks
\begin{align}
	\min_{\eta_t, \mathbf{x}_t} \quad &\eta_t, \label{eq: epi} \\
	\textrm{~~s.t.}
	\quad & f_{i,t}(x_{i,t})\leq\eta_t, \forall i \in \mathcal{N},~\label{eq:eta}\\
	& \sum_{i\in\mathcal{N}}x_{i,t}\leq 1,~\label{eq:sum_}\\
	&x_{i,t}\geq 0, \forall i \in \mathcal{N}.~\label{eq:x_i_}
\end{align}
\endgroup
We can further penalize the constraints (\ref{eq:eta}) with some multiplier $r_{i,t}>1$. An equivalent form of problem (\ref{eq: epi}) can be obtained as follows (\cite{SN13}):
\begingroup
\allowdisplaybreaks
\begin{align}
	\min_{\eta_t, \mathbf{x}_t} \quad &\sum_{i\in\mathcal{N}}\frac{\eta_t}{N}+r_{i,t}[f_{i,t}(x_{i,t})-\eta_t]^{+}, \label{eq:cost} \\
	\textrm{~~s.t.}
	\quad & (\ref{eq:sum_} )-(\ref{eq:x_i_}).\nonumber
\end{align}
\endgroup
where $[\cdot]^{+} =\max\{0, \cdot\}$.
An important property of problem (\ref{eq:cost}) is that $r_{i,t}$ can be \textit{arbitrarily} chosen by each agent independently of the others, which allows distributed implementation.

However, we note that (\ref{eq:cost}) is not continuously differentiable. 
If we used it over time to construct a conventional online optimization, we would face the same challenges as solving our original online min-max problem. 
Therefore, we seek a new solution. 
Next, we will present how problem (\ref{eq:cost}) can be used as the starting point to build the proposed \texttt{DORA} algorithm.
\subsection{Distributed Online Optimization}
A naive solution to problem (\ref{eq:cost}) is to use Lagrange decomposition.
By penalizing the coupling constaint (\ref{eq:sum_}) with some partial Lagrange multiplier $\pi_t\geq 0$, the partial Lagrangian of (\ref{eq:cost}) is
\begin{equation}
	\begin{aligned}\label{eq: lagrangian}
		L_t(\eta_t, \mathbf{x}_t, \pi_t)=\sum_{i\in\mathcal{N}}\left(\frac{\eta_t}{N}+r_{i,t}[f_{i,t}(x_{i,t})-\eta_t]^{+}\right)
		+\pi_t\left(\sum_{i\in\mathcal{N}}x_{i,t}- 1\right).
	\end{aligned}
\end{equation}
For convex problems, strong duality holds under Slater's condition.
However, even for convex problems, it is unclear how to find the optimal $\pi_t^*$.

To overcome this obstacle, we first observe in the following lemma that a \textit{local} version of the Lagrangian admits a closed-form minimizer for any choice of $\eta_t$ in the range of $f_t(\mathbf{x})$ for $\mathbf{x}$ satisfying (\ref{eq:sum_}) and (\ref{eq:x_i_}), even though the optimal $\pi_t^*$ is unknown.
Let $x'_{i,t}$ be an arbitrary solution to $ f_{i,t}(x'_{i,t})=\eta_t$, i.e., 
\begin{align}\label{eq: x'}
	x'_{i,t}
	\in   f_{i,t}^{-1}(\eta_t).
\end{align}
Since $ f_{i,t}$ is a monotonic function, root $x'_{i,t}$ can be found efficiently by bisection (\cite{HP76}).
As we will see later, this is an important quantity to our algorithm design. Further details are given in Remark \ref{rm: 1}. 
For now, in the following lemma, we show that $x'_{i,t}$ is a minimizer of the local Lagrangian.
Its proof can be found in Appendix A.

\begin{lemma}\label{lemma: 1}
	Suppose $f_{i,t}$ is convex for all $i$ and $t$. 
	Let $(\mathbf{x}_t^*,\eta_t^*,\pi_t^*)$ denote a minimizer of (\ref{eq: lagrangian}). 
	Then, for any choice of $\eta_t$ in the range of $f_t(\mathbf{x})$ for $\mathbf{x}$ satisfying (\ref{eq:sum_}) and (\ref{eq:x_i_}), $x'_{i,t}$ is a minimizer of the local Lagrangian
	\begin{equation}\label{eq: local}
		\begin{aligned}
			L_{i,t}(x_{i,t})=r_{i,t}[f_{i,t}(x_{i,t})-\eta_t]^{+}+\pi_{t}^* x_{i,t}.
		\end{aligned}
	\end{equation}
\end{lemma}

	
	Unfortunately, if each agent updates its $x_{i,t}$ based on (\ref{eq: x'}), the resultant solution is still not optimal to problem (\ref{eq:cost}).
	This is because the global cost can be further reduced by simultaneously allocating the remaining resource budget to each agent.
	Interestingly, a similar phenomenon has been observed in general unconstrained sum minimization problems (\cite{LY2020}).
	This inspires the proposed \texttt{DORA} algorithm, which will be shown in Section \ref{sec: analysis} to provide bounded dynamic regret for our online min-max problem.

	The intuition of \texttt{DORA} is to make fast agents relinquish an appropriate amount of resource and give it to the agent who was the straggler in the \textit{previous} online round.
	Thus, in each round, the agents who are non-stragglers in the previous online round move towards the minimizer of their local Lagrangian in (\ref{eq: local}), while the stragglers will be allocated the remaining resource budget.

\algrenewcommand\algorithmicrequire{\textbf{Agent $i=1,2,\ldots, N$ \textbf{runs in parallel:}}}
\algrenewcommand\algorithmicensure{\textbf{Parameter server runs:}}

\begin{algorithm}

		\caption{Distributed Online Resource Re-Allocation}\label{alg:online_alg}
		\SetAlgoLined
	\textbf{Input:} Number of rounds $T$, and step size $\alpha$.\\
	\textbf{Initialization:} Arbitrary initial allocation $\mathbf{x}_{0}\in\mathcal{F}$. \\
	\textbf{Agent $i=1,2,\ldots, N$ \textbf{runs in parallel:}}
	
		 \For{round $t=1,2,..., T$}{
	 Observe $f_{i,t-1}(\cdot)$; \Comment{Delayed information on cost} \label{alg: observe}\\
		 Receive $f_{t-1}(\mathbf{x}_{t-1})$ from the parameter server; \label{alg: receive}\\
		 Compute $x_{i,t}$ using (\ref{eq: updateA}); \Comment{Resource relinquishment}\label{alg: update}\\
		Send $x_{i,t}$ to the parameter server;\label{alg: syn}
}

\textbf{Parameter server runs:}

			\For{round $t=1,2,..., T$}{
 Receive $x_{i,t}$ from agent $i$; \label{alg: aggregate}\\
 Update $x_{s_t,t}$ using (\ref{eq: updateB}); \label{alg: update_} \Comment{Resource re-allocation}\\
 Allocate resource based on $\mathbf{x}_{t}$;\label{alg: alloc} \\
 Observe and send $f_{t}(\mathbf{x}_t)$ to all agents;\label{alg: reveal_}
}
\end{algorithm}

	The pseudocode of \texttt{DORA} is shown in Algorithm \ref{alg:online_alg}.
After observing the local cost function $f_{i,t-1}(\cdot)$ in the previous round and receiving the feedback $f_{t-1}(\mathbf{x}_{t-1})$ from the parameter server, each agent adjusts its decision by moving towards the minimizer of its local Lagrangian with any step size $0<\alpha<1$, i.e., 
\begin{equation}\label{eq: updateA}
	\begin{aligned}
		x_{i,t}\leftarrow x_{i, t-1}-\alpha (x_{i,t-1}-x'_{i, t-1}), \forall i \in\mathcal{N},
	\end{aligned}
\end{equation}
where $x'_{i,t-1}$ is a minimizer of (\ref{eq: local}) in round $t-1$ and can be computed as in (\ref{eq: x'}).
Afterward, the local decisions are synchronized with the parameter server. 
After receiving $x_{i,t}$ from each agent, the resource relinquished by non-stragglers is re-allocated to the agent who is the straggler in the previous round:
\allowdisplaybreaks\begin{align}
	x_{s_{t-1},t}\leftarrow 1-\sum_{i\neq s_{t-1}}x_{i,t}
	= x_{s_{t-1},t-1}-\alpha\sum_{i\neq s_{t-1}}(x_{i,t-1}'-x_{i, t-1}),\label{eq: updateB}
\end{align}
where $s_{t-1}=\arg\max_{i\in\mathcal{N}}f_{i,{t-1}}(x_{i,{t-1}})$ denotes the straggler in round ${t-1}$.
Then, the parameter server observes the value of the global cost incurred in this round and sends it to all agents.
It can be easily verified that constraint (\ref{eq:sum}) is always satisfied in re-allocation since the relinquished resource is shared among stragglers.
\newtheorem{rremark}{Remark}
\begin{rremark}\label{rm: 1}
	$x'_{i,t-1}$ can be interpreted as the minimum resource agent $i$ needs to maintain such that $f_{i,t-1}(x'_{i,t-1})\leq f_{t-1}(\mathbf{x}_{t-1})$.
	Thus, $(x_{i,{t-1}}-x'_{i, {t-1}})$ represents the maximum resource the agent can relinquish without making itself the straggler based on the historical loss function in round $t-1$.
	Furthermore, $\alpha(x_{i,{t-1}}-x'_{i, {t-1}}),\forall \alpha\in[0,1]$ can be viewed as the amount of resource that agent $i$ chooses to relinquish and share with the straggler.
	It is easy to verify that the straggler does not relinquish its resource at this step since $x_{s_{t-1},{t-1}}=x'_{s_{t-1}, {t-1}}$.
\end{rremark}

\begin{rremark}
	The rationale behind the step size $\alpha<1$ in (\ref{eq: updateA}) is that, for instantaneous optimization, if each agent maintains the minimum resource $x'_{i,t}$, the total resource relinquished to the straggler $s_t$ could exceed its optimal allocation $x^*_{s_t,t}$. 
	Then the total resource shared by the non-stragglers would exceed the straggler's needs. 
	The detailed proof is presented in the next section.
	Rather than directly minimizing (\ref{eq: local}), the agents who are non-stragglers in the previous round take a step toward the minimum with step size $\alpha$.
\end{rremark}

From Algorithm \ref{alg:online_alg}, we see that \texttt{DORA} has several advantages in terms of its ease of implementation: 1) it avoids gradient calculation and projection onto the feasible set, 2) each agent only needs to keep its local variable rather than a full copy of all decision variables, and 3) communication only needs to be maintained between the agents and the server.

We emphasize here that the updates in (\ref{eq: updateA}) and (\ref{eq: updateB}) are different from gradient descent. 
In fact, since each agent only needs to compute $x'_{i,t-1}$ in each round, independently of the other agents, the overall computation complexity of DORA is $O(N)$ per round. 
In comparison, applying projected gradient-based online algorithms would lead to $O(N^2)$ computation complexity per round for gradient calculation alone (\cite{BV04}). 
Furthermore, since the computation complexity of projection is $O(N \log N)$ in Euclidean space~(\cite{LY09}), the overall computation complexity would be at least $O(N^2 \log N)$ per round.
Therefore, the $O(N)$ complexity of DORA is a substantial improvement. 
%
	
	\section{Dynamic Regret Analysis}\label{sec: analysis}
	Besides ease of implementation, another main advantage of \texttt{DORA} is that it provides strong performance guarantees in the form of bounded dynamic regret. 
	For dynamic regret analysis, we make the following assumptions: 
	\newtheorem{assumption}{Assumption}
	\begin{assumption}[Monotonicity]\label{eq: cvx}
		If $x\leq y$, then $f_{i,t}(x)\geq f_{i,t}(y),\forall i \text{~and~} t$.
	\end{assumption}
	\begin{assumption}[Lipschitz]\label{eq: gradient}
		$|f_{i,t}(x)-f_{i,t}(y)|\leq L ||x-y||,\forall i \text{~and~} t$.
	\end{assumption}
	
	\begin{theorem}\label{Thm: regret}
		Consider the online min-max problem defined in (\ref{eq: minmax}) with Assumptions \ref{eq: cvx} and \ref{eq: gradient}.
		With a fixed step size $\alpha$, the dynamic regret $\textnormal{Reg}^d_T$ for the sequence of decisions $\mathbf{x}_t$ generated by \texttt{DORA} is upper bounded by
			$\textnormal{Reg}^d_T\leq \sqrt{TL^2(\frac{3}{2\alpha}+\frac{P_T}{\alpha}+\frac{T(4+\alpha)}{2})}$,
		where $P_T$ is the path-length of the dynamic minimizers.
	\end{theorem}
	\begin{proof}
		Please refer to Appendix B.
	\end{proof}

\section{Application to Distributed Learning  in Mobile Edge Computing}\label{ch: performance}

To complement the theoretical findings in the previous section, we numerically study the performance of \texttt{DORA}, for an application to online resource allocation in distributed edge learning. 
Edge learning benefits from geographic proximity to where data are generated and processed. Recent surveys on distributed learning over the wireless edge can be found in~\cite{AG20} and \cite{ND20}. 
Here, we focus on online bandwidth allocation for the communication between a set of agents and a parameter server, to reduce the training latency. 

\subsection{Online Min-Max Resource Allocation for Edge Learning Using \texttt{DORA}}

We consider a parameter server-based distributed learning framework, as shown in Fig.~\ref{fig:system}, where an edge server is responsible for resource allocation and necessary coordination. 
Multiple agents $i\in\mathcal{N}$ train the same learning model in parallel with their local training datasets, while the edge server also functions as the parameter server.

The computation task associated with agent $i$ in round $t$ has a data size of $d_{i,t}$.
For spectrum sharing among agents, we assume orthogonal frequency division. 
Therefore, the wireless communication delay for agent $i$ is given by $f_{i,t}^{\text{C}}=d_{i,t}/(x_{i,t} B\log(1+\frac{h_{i,t}^2 p_{i,t}}{\sigma_n^2}))$, where we have used the Shannon bound for the data rate, $x_{i,t}$ is the fraction of bandwidth allocated to agent $i$, $B$ is the total available bandwidth, $h_{i,t}^2$ is the channel power gain, $p_{i,t}$ is the transmit power, and $\sigma_n^2$ is the white noise power. 

The total delay due to agent $i$ consists of both the communication delay and the processing delay, i.e.,
\begin{equation}\nonumber
	\begin{aligned}
		f_{i,t}(x_{i,t}) = f_{i,t}^{\text{C}}+f_{i,t}^{\text{P}},
	\end{aligned}
\end{equation}
where $f_{i,t}^{\text{P}}$ is the time required for agent $i$ to process its task in round $t$. 
Note that $f_{i,t}^{\text{P}}$ depends on the computation intensity of the task and the available processing capacity of the agent, both of which can be time-varying and unpredictable. 
We do not require knowledge of $f_{i,t}^{\text{P}}$ at the beginning of round $t$.

\begin{figure}
	\centering
	\includegraphics[width=0.6\linewidth]{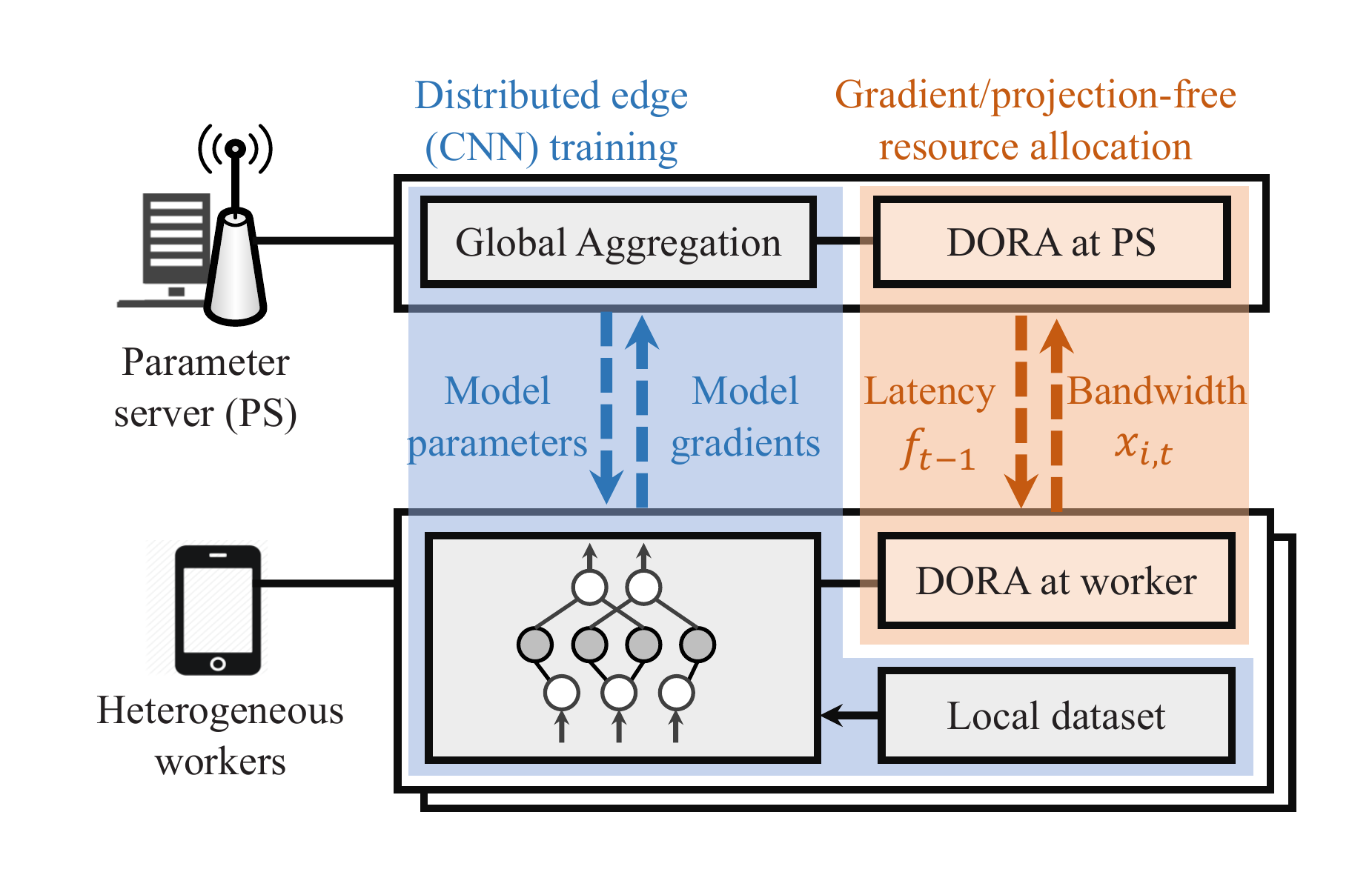}
	\caption{\texttt{DORA} for distributed learning at wireless edge.}
	\label{fig:system}
\end{figure}

We consider standard synchronous distributed learning, where the parameter server updates the models only after aggregating all local gradients. 
Therefore, the latency of each round depends on the agent with the highest latency.
We are interested in minimizing the maximum latency per round over time. 
Therefore, we define the cost function $f_t(\mathbf{x}_t)$ as in (\ref{eq: max}).
Then we can formulate resource allocation in distributed edge learning as an online min-max optimization problem as in (\ref{eq: minmax}).

As illustrated in Fig.~\ref{fig:system}, using convolutional neural network (CNN) training as an example, the resource allocation scheme provided by \texttt{DORA} can be integrated into the distributed edge learning process as follows.
Set the round index $t$ of \texttt{DORA} to be the same as that of CNN training.
After each agent calculates the local gradient of CNN in round $t-1$, the processing time $f_{i,t-1}^{\text{P}}$  is observed by the agent.
After each agent sends its local gradient to the parameter server in round $t-1$, the communication time $f_{i,t-1}^{\text{C}}$, data size $d_{i,t}$, and channel condition are observed by the agent.
That is to say that the form of the latency function $f_{i,t-1}$ becomes known to the agent at the beginning of round $t$ (line \ref{alg: observe} in Algorithm~\ref{alg:online_alg}).
After aggregating the local gradients in round $t-1$, the parameter server broadcasts the updated CNN model as well as the latency $f_{t-1}$ in the previous round to all agents.
Then each agent determines the bandwidth usage in round $t$ based on the feedback from the parameter server.
Afterward, the parameter server re-allocates the bandwidth among all agents.

\subsection{Experimental Performance Evaluation}

We consider a distributed learning scenario where five agents train the LeNet model in parallel, using the MNIST database~(\cite{LB98}).
Our learning system is implemented with the distributed package (i.e., torch.distributed) in PyTorch.
The hardware used for the experiments features a 2.9 GHz Intel Core i5 processor and 8 GB of memory. 

LeNet is a CNN composed of two convolutional layers, followed by two fully-connected layers and a softmax classifier. 
In total, we have 61,706 trainable parameters with a size of $d_{i,t}=0.35$ MB, $\forall i,t$. 
We train LeNet with the cross-entropy loss and the Adam optimizer. 
The learning rate of LeNet model training is set to 0.001.
The MNIST database has a training set of 60,000 handwritten digits images, each of which contains 28 x 28 greyscale pixels.
It is equally partitioned into five local training datasets, each assigned to an agent with a batch size of 256.
Therefore, we have $\lceil\frac{60,000}{5\times 256}\rceil=47$ rounds in each training epoch.
We set the number of training epochs to $10$, i.e., $470$ rounds.
We measure and use the actual processing time $f_{i,t}^{\text{P}}$ for training at each agent in each round on the processor. 

The parameter server is placed at the center of a 500 m $\times$ 500 m area.
The agent mobility model follows the random waypoint model where the initial agent velocity is chosen uniformly at random in the interval $[0.8v,1.2v]$.
The bandwidth resource budget $B$ is $20$ MHz.
The channel power gain $h_{i,t}^2$ is generated as  $h_0 (\frac{D_0}{D_i})^n$, where $h_0 = - 40$ dB is the path-loss constant, $D_0= 1$m is the reference distance, $D_i$ is the distance between agent $i$ and the parameter server, and $n=4$ is the path-loss component~(\cite{MZ16}).
The transmission power $p_{i,t}=1$ W, $\forall i,t$.
The noise density is $-174$ dBm/Hz.
We compare the performance of \texttt{DORA} with that of the following gradient- and/or projection-based benchmarks:

\begin{itemize}
	\item \textbf{EQUAL}: Resource is equally shared among all edge devices. This is frequently assumed in the analysis of distributed training.
	\item \textbf{OGD-OMM}~(\cite{BC19}): 
	Resource allocation is updated with gradient descent, i.e., $\mathbf{x}_{t+1}\leftarrow \pi_\mathcal{F}\left(\mathbf{x}_{t}-\alpha \mathbf{g}_t\right)$,
	where $\mathbf{g}_t$ can be a subgradient when the gradient does not exist, and $\pi_\mathcal{F}(\cdot)$ denotes the Euclidean projection onto $\mathcal{F}$.
	Projection onto the simplex is implemented using the method in (\cite{MA14}).

	\item \textbf{OMD}~(\cite{SJ17}): 
	Resource allocation is updated with mirror descent, i.e., $\mathbf{x}_{t+1}\leftarrow \arg\min\left\{\langle \mathbf{x},g_t\rangle+\alpha D_\psi(\mathbf{x},\mathbf{x}_t)\right\}$, where $D_\psi(\mathbf{x}, \mathbf{y})=\psi(\mathbf{x}) - \psi(\mathbf{y}) - \langle \mathbf{x} - \mathbf{y}, \nabla\psi(\mathbf{y})\rangle$ is the Bregman divergence corresponding to a strongly convex regularization function $\psi$. 
	Here we consider the Kullback-Leibler (KL) divergence, which is generated by $\psi(\mathbf{x})=\sum_{i\in\mathcal{N}} x_i \log x_i$.
	\item \textbf{FKM}~(\cite{FL05}): Resource allocation is updated with estimated gradient descent: perturb direction with $\mathbf{v}_{t}\leftarrow \mathbf{x}_t+\delta_t \mathbf{u}_t$, where $\delta_t$ is a step size and $\mathbf{u}_t$ is a randomly selected unit vector; estimate gradient with $\hat{\mathbf{g}}_{t}\leftarrow \frac{N}{\delta_t}f_t(\mathbf{v}_{t})\mathbf{u}_t$; and then update decisions with $\mathbf{x}_{t+1}\leftarrow \pi_{\mathcal{F}}\left(\mathbf{x}_{t}-\alpha \hat{\mathbf{g}}_t\right)$. 
	
	\item \textbf{OCG}~(\cite{ZZ17}): Resource allocation is updated with conditional descent, i.e.,
	$\mathbf{v}_{t}\leftarrow \arg\min_{x\in\mathcal{F}}\langle \mathbf{x},\partial F_t(\mathbf{x}_{t})\rangle$ and  $\mathbf{x}_{t+1}\leftarrow \mathbf{x}_{t}+ \frac{1}{t}\left(\mathbf{v}_{t}-\mathbf{x}_{t}\right)$, where $\partial F_{t}(\mathbf{x}_{t}) = \sum_\tau^{t}\partial f_\tau(\mathbf{x}_{\tau})$ is the aggregated gradient.

	\item \textbf{Dynamic OPT}: Ideally, we assume a priori knowledge of all system variables, and we solve the instantaneous optimization problem in each round through Sequential Quadratic Programming. 
	This is also the comparator in the definition of dynamic regret.
\end{itemize}

All algorithms are initialized with equal bandwidth allocation.
We use a fixed step size $\alpha = 0.02$ for resource allocation schemes \texttt{OGD-OMM}, \texttt{OMD}, and \texttt{DORA}.

Figs.~\ref{fig: d_latency} and \ref{fig: d_regret} compare the performance of various algorithms over the number of rounds $T$ in terms of per-round latency and the dynamic regret, respectively. 
Here we set the velocity of users as $v=0$.
Note that, besides serving as the optimum baseline, \texttt{Dynamic OPT} also provides an indication on how fast the environment fluctuates.
We observe that \texttt{DORA} tracks closely to \texttt{Dynamic OPT} after the first few rounds and has the least dynamic regret.
\texttt{EQUAL} incurs the worst latency since it ignores the heterogeneity among agents as well as the time-varying cost.
\texttt{FKM}, \texttt{OMD}, \texttt{OGD-OMM}, and \texttt{OCG} gradually track close to \texttt{Dynamic OPT} as time goes on, but they require many more rounds than \texttt{DORA}.
\texttt{OGD-OMM} and \texttt{OMD} adjust their decisions only using parameters relevant to the historical communication delay since the gradient of function $f_{i,t}$ only depends on the communication parameters.
\texttt{FKM} converges more slowly than \texttt{OGD-OMM} since it depends on the accuracy of gradient estimation.
For \texttt{OCG}, the decision is updated through a trade-off between the current decision and the output of the linear optimization, which allocates all resource to the straggler given the resource constraints.
Therefore, non-stragglers in the current round could become stragglers in the next round, thus leading to frequent latency spikes.
In contrast, in \texttt{DORA}, the non-stragglers relinquish some resource to the agent who is the straggler in the previous round, so it updates its decisions by jointly considering the historical computation and communication latency.
As shown in Fig.~\ref{fig: d_latency}, by round $100$, \texttt{DORA} has reduced the per-round latency by 66\%, 65\%, 60\%, 45\%, and 27\%, respectively, when compared with \texttt{EQUAL}, \texttt{FKM}, \texttt{OGD-OMM}, \texttt{OMD}, and \texttt{OCG}.

\begin{figure*}[t]
	\centering
	\begin{minipage}[b]{.45\textwidth}
		\centering
 		\includegraphics[width=0.8\linewidth]{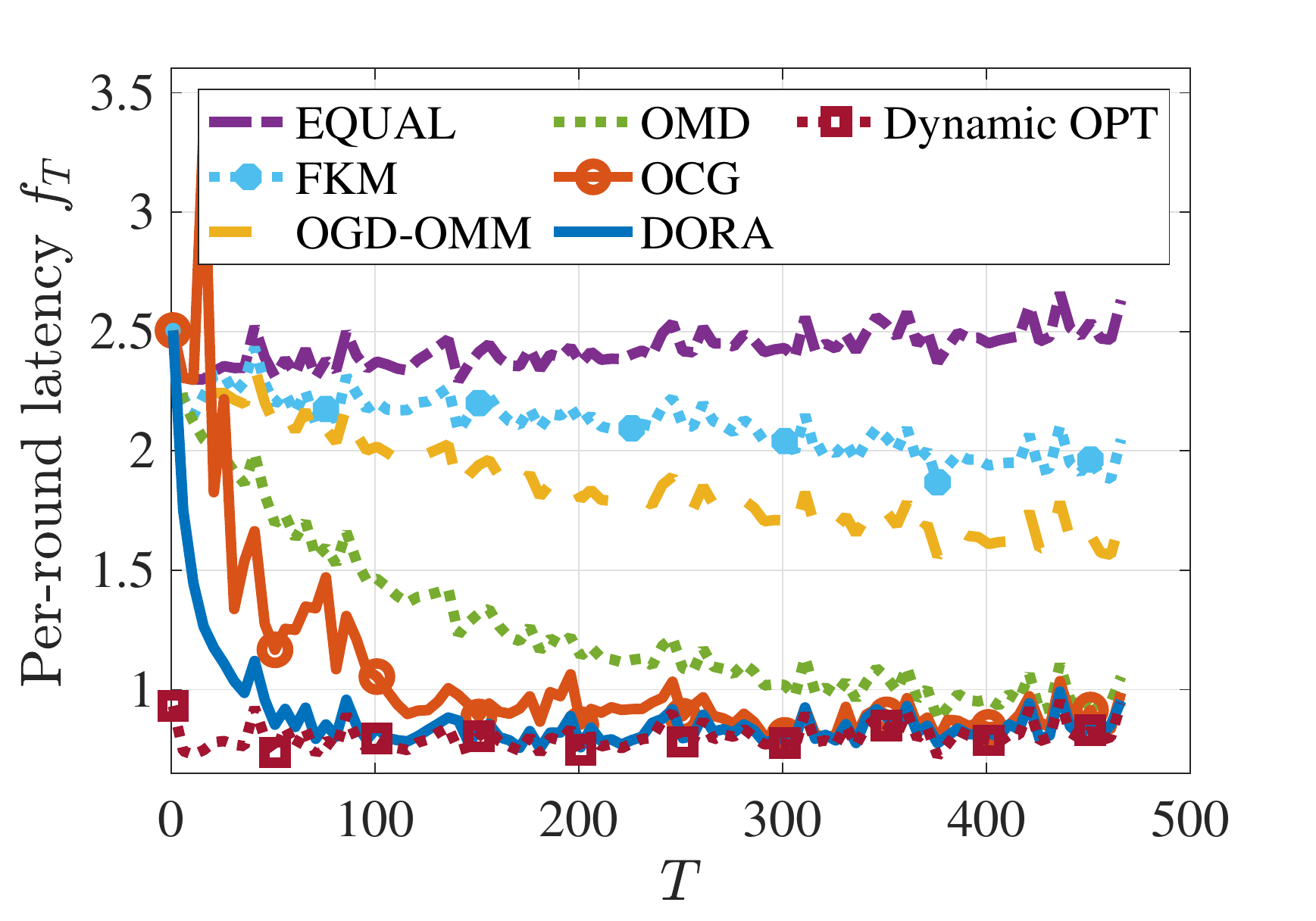}
		\caption{Per-round delay.}
		\label{fig: d_latency}
	\end{minipage}
	\begin{minipage}[b]{.45\textwidth}
		\centering
		\includegraphics[width=0.8\linewidth]{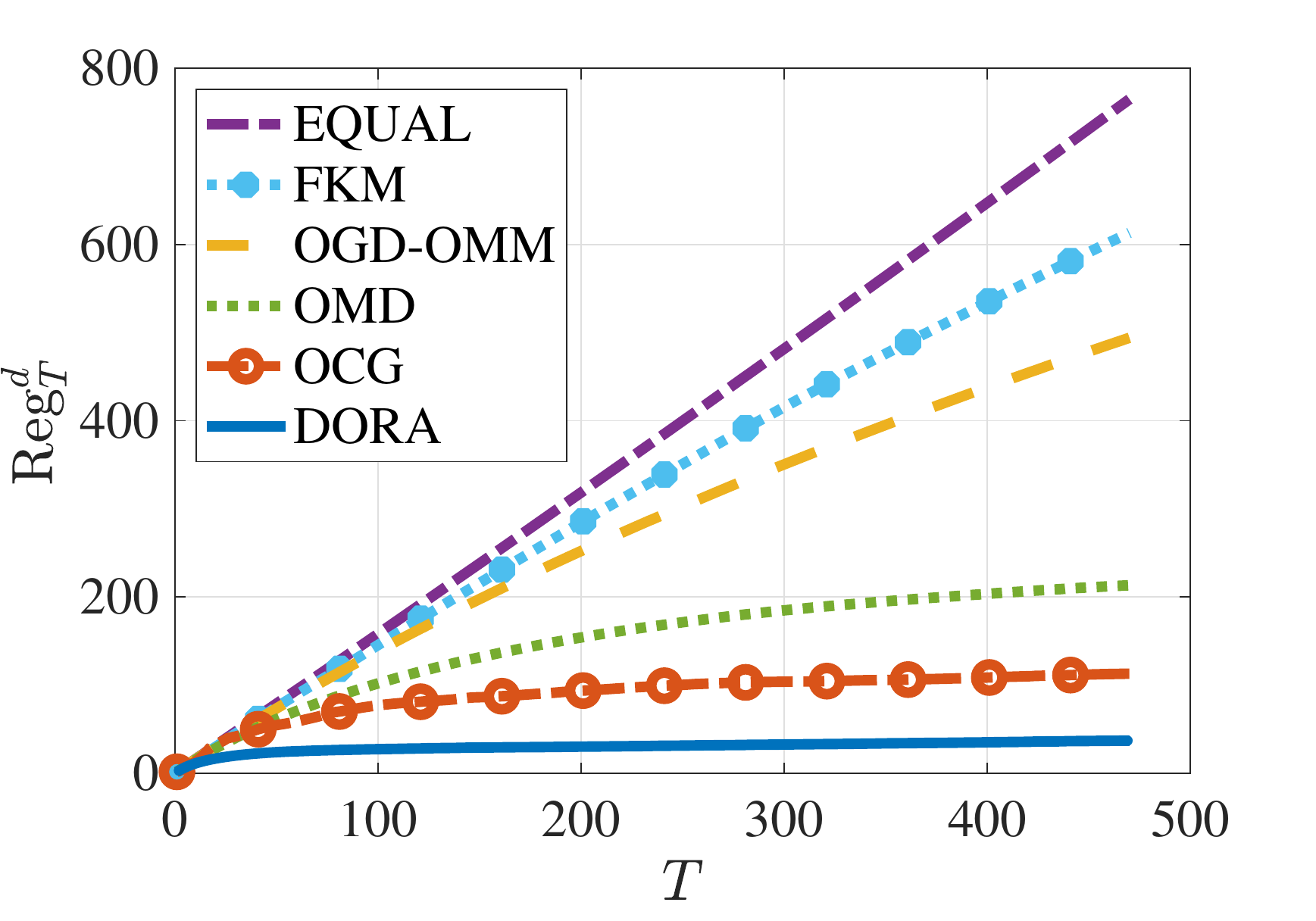}
		\caption{Dynamic regret.}
		\label{fig: d_regret}
	\end{minipage}
\end{figure*}  

\begin{figure*}[t]
	\centering
	
	\begin{minipage}[b]{.45\textwidth}
		\centering
		\includegraphics[width=0.8\linewidth]{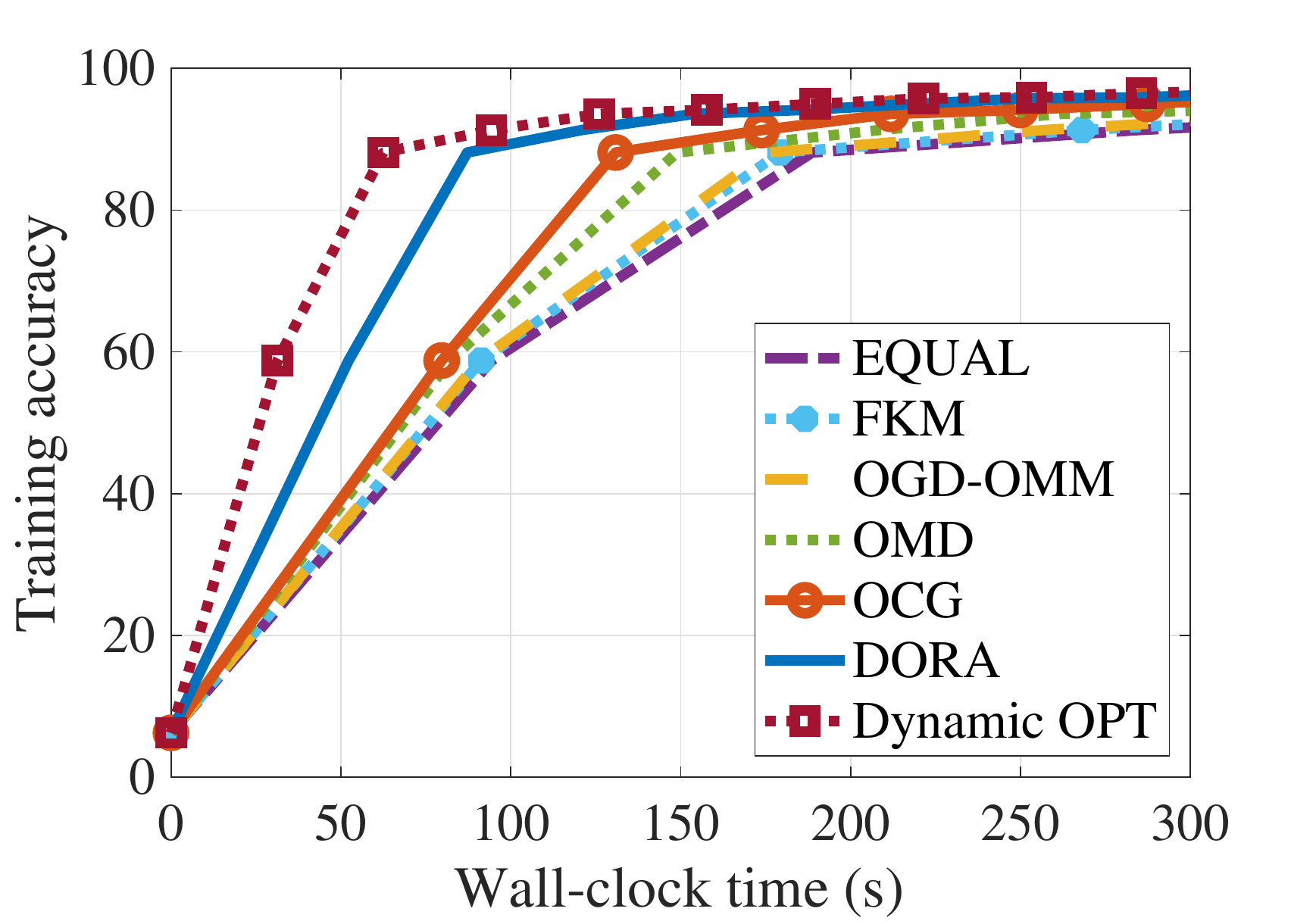}
		\caption{Training accuracy.}
		\label{fig: training}
	\end{minipage}%
	\begin{minipage}[b]{.45\textwidth}
		\centering
		\includegraphics[width=0.8\linewidth]{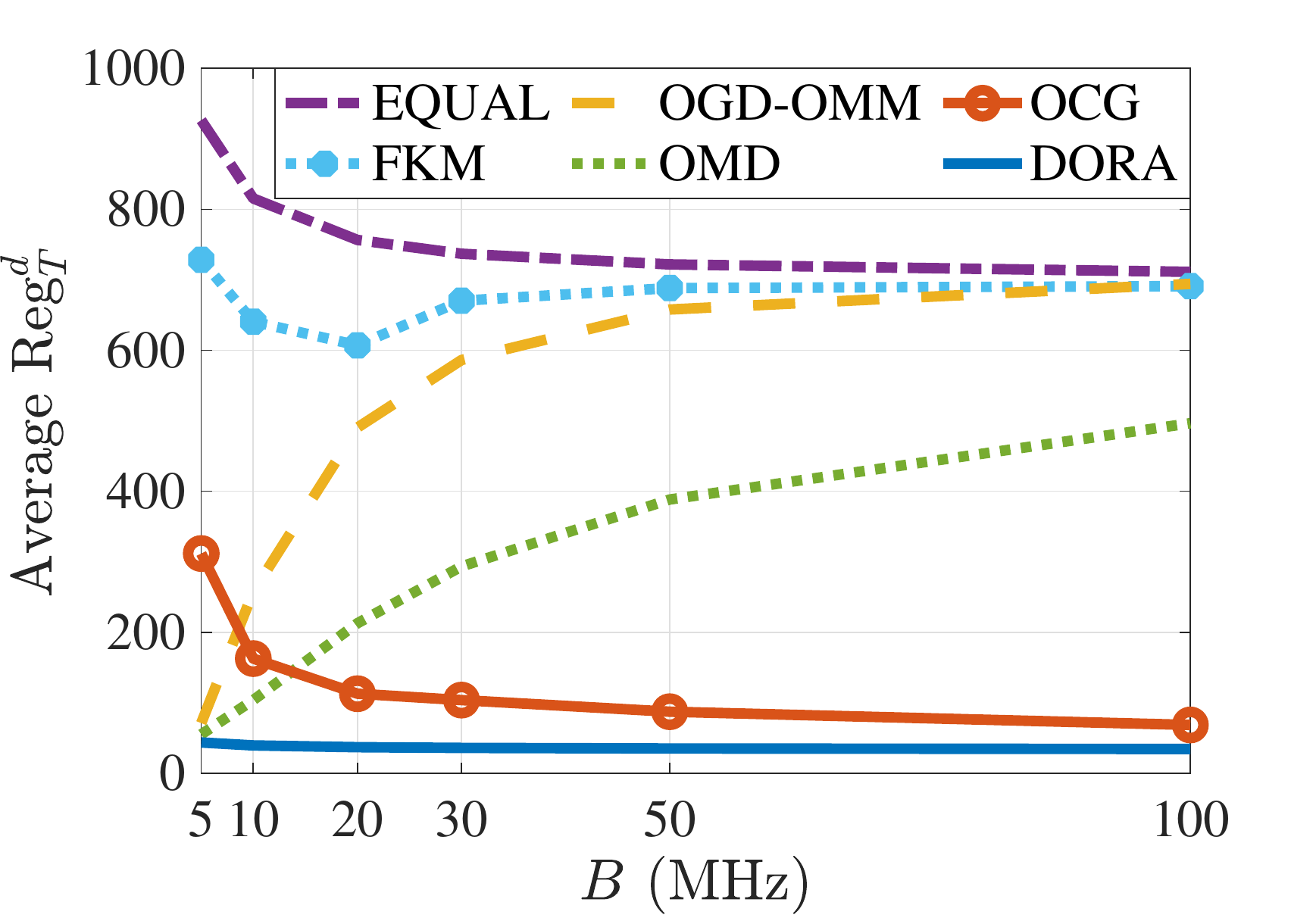}
		\caption{Average regret over bandwidth.}
		\label{fig:c2cratio}
	\end{minipage}
\end{figure*}  

\begin{figure*}[t]
	\centering
	\begin{minipage}[b]{.45\textwidth}
		\centering
		\includegraphics[width=0.8\linewidth]{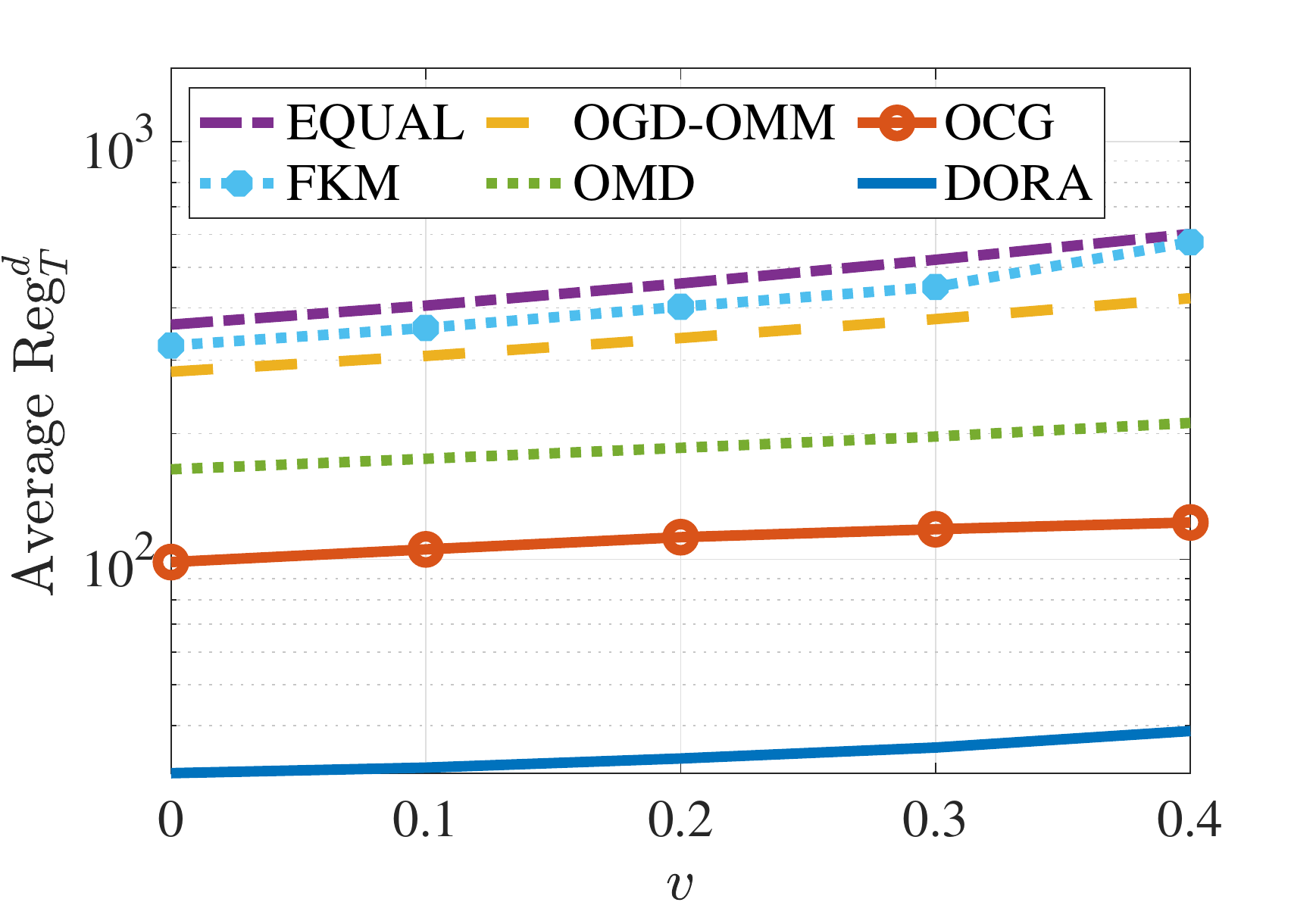}
		\caption{Average regret over agent velocity.}
		\label{fig:randomness}
	\end{minipage}
	\begin{minipage}[b]{.45\textwidth}
		\centering
		\includegraphics[width=0.8\linewidth]{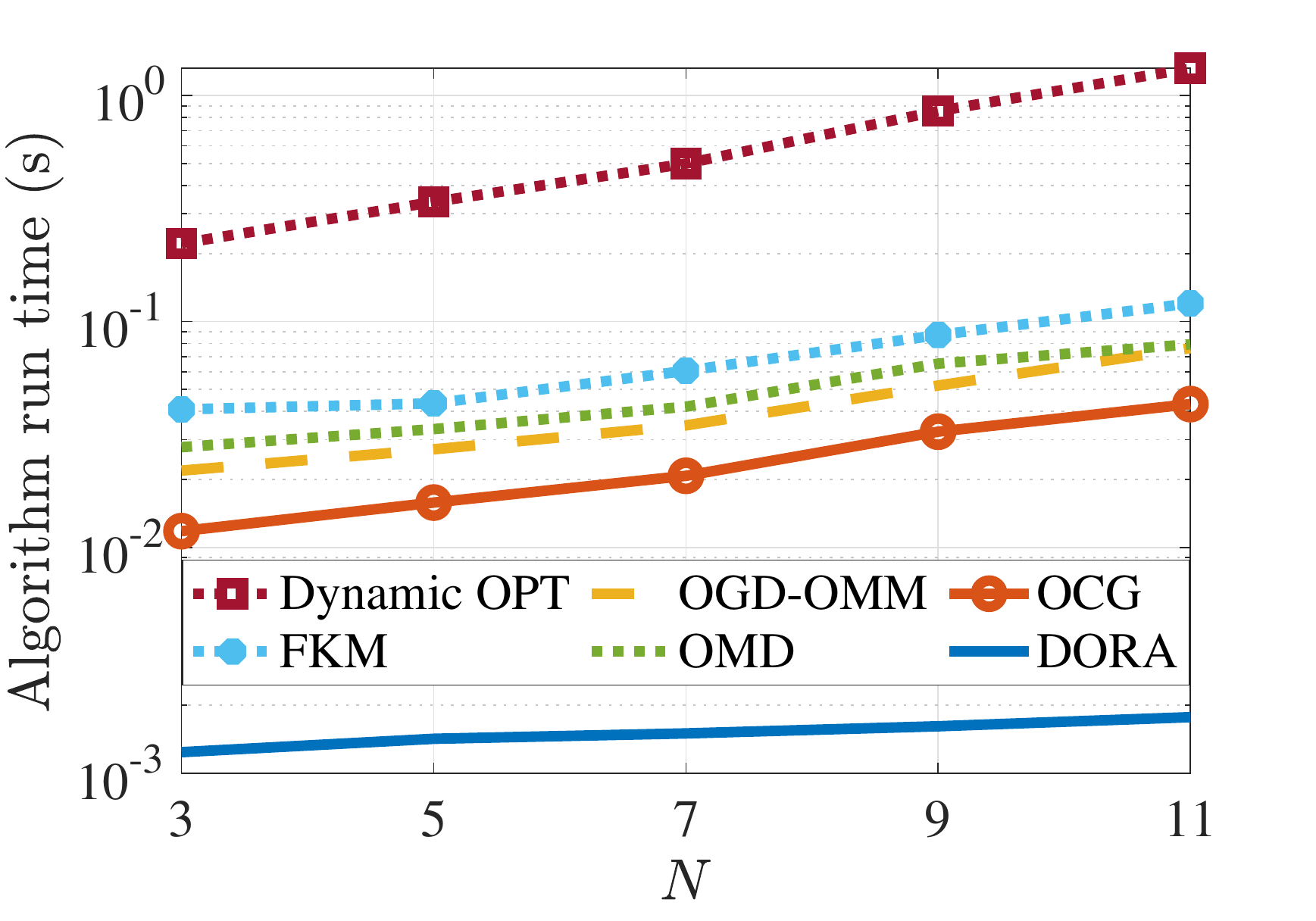}
		\caption{Algorithm run time over $N$.}
		\label{fig:algrunning}
	\end{minipage}
	
\end{figure*}

Fig.~\ref{fig: training} shows the training accuracy of LeNet versus the training (wall-clock) time.
Recall that the LeNet model is updated only after all agents send their local gradients to the parameter server. 
The training time consists of both the computation time for LeNet's forward and backward propagation and the communication time for parameter synchronization, which accounts for both the number of rounds and the latency per round. 
We observe that \texttt{DORA} has the fastest training time since it responds quickly towards the dynamic optimum.
For 90\% training accuracy, \texttt{DORA} speeds up the training time by 53.7\%, 51.4\%, 50.3\%, 41.2\%, and 34\%, respectively, when compared with \texttt{EQUAL}, \texttt{FKM}, \texttt{OGD-OMM},  \texttt{OMD}, and \texttt{OCG}.

We further investigate the performance of different algorithms under various bandwidth budgets.
More bandwidth indicates a network with a higher communication speed.
To measure the converged performance of resource allocation algorithms, we use the average of $\textnormal{Reg}_T^d$ over $460 \leq T \leq 470$ to approximate the average regret.
Fig.~\ref{fig:c2cratio} shows the impact of the bandwidth budget on the average regret.
For gradient-based algorithms, the average regret increases since the magnitude of the gradient decreases with increasing allocated bandwidth.
This indicates that they converge more slowly in higher-speed networks.
It is worth noting that \texttt{FKM} does not work well with small $B$.
This is because the gradient estimation with a unit vector for small $B$ becomes inaccurate.
Moreover, since \texttt{OCG} blindly shows partiality to the straggler, it is likely to introduce frequent latency spikes in lower-speed networks.
In contrast, \texttt{DORA} updates its decisions by jointly considering the latency in the previous round as well as the heterogeneity among agents, ensuring robust performance under different bandwidth budgets.

Fig.~\ref{fig:randomness} shows the impact of agent velocity $v$, which represents the randomness of the communication link, on the average regret.
A large $v$ indicates more radical changes in the wireless channels, and thus more radical changes in the local cost functions over time.
The average regret of all algorithms increases with increasing $v$.
This is in line with the theoretical result that a larger path-length $P_T$ of the instantaneous minimizers leads to a larger regret.
In particular, with larger $v$, the previous feedback becomes less useful as a reference and the dynamics in communication delay gradually overweigh the computation heterogeneity among the agents.
In a steady environment, i.e., $v = 0$, \texttt{DORA} reduces the average regret by 91.8\%, 90.7\%, 89.2\%, 81.3\%, and 69\%, respectively, when compared with \texttt{EQUAL}, \texttt{FKM}, \texttt{OGD-OMM}, \texttt{OMD}, and \texttt{OCG}.

To demonstrate the computation complexity of \texttt{DORA} and the other algorithms, we measure their run time (in seconds) for calculating the resource allocation solutions only, i.e., it excludes the LeNet training time.
Since \texttt{EQUAL} is deterministic, its run time is zero and is ignored in this comparison.
Fig.~\ref{fig:algrunning} shows the total run time over the horizon of $T=470$ rounds, for a different number of agents $N$.
With increasing $N$, the run time of \texttt{Dynamic OPT} and \texttt{OCG} increases dramatically since solving an instantaneous optimization with the standard optimization package includes multiple inner loops.
\texttt{FKM}, \texttt{OGD-OMM}, and \texttt{OMD} are also sensitive to $N$ due to their projection operation.
In contrast, the algorithm run time of \texttt{DORA} is negligible.
\texttt{DORA} is robust and light-weight due to its gradient-free and projection-free properties.

\section{Conclusion}\label{sec: conclusion}
We have proposed a new distributed online algorithm termed \texttt{DORA} to solve an online min-max optimization problem in a multi-agent system with coupling linear constraints.
\texttt{DORA} has high computation efficiency, since it is gradient-free and projection-free.
We analyze the dynamic regret of \texttt{DORA} for non-convex functions.
We applied the proposed solution to resource allocation in distributed online machine learning.
Numerical studies show the efficacy of the proposed solution in terms of significantly faster convergence and reduced training time.

	\acks{This work has been funded by the Natural Sciences and Engineering Research Council (NSERC) of Canada.}

\bibliography{DORA}

\appendix
\section{Proof of Lemma \ref{lemma: 1}}\label{proof: l1}
	\begin{proof}
	By the definition of $\eta_t^*$, we have $\eta_t \geq \eta_t^*$, i.e., $f_{i,t}(x'_{i,t})\geq f_{i,t}(x_{i,t}^*)$.
	As $f_{i,t}(x_{i,t})$ is a decreasing function over $x_{i,t}$, we have $x'_{i,t}\leq x_{i,t}^*$.
	With the KKT condition, at optimality of minimizing (\ref{eq: lagrangian}), we have
	\begin{equation}\nonumber
		\begin{aligned}
			0 \in \frac{\partial 	L_t(\eta_t^*, \mathbf{x}_t, \pi_t^*)}{\partial x_{i,t}}|_{x_{i,t}^*} 
			=r_{i,t}\partial f_{i,t}(x_{i,t}^*)+\pi_t^*, ~\forall i\in\mathcal{N}, 
		\end{aligned}
	\end{equation}
	where $\partial f_{i,t}(x_{i,t}^*)$ is the set of subgradients of $f_{i,t}$ at $x_{i,t}^*$.
	Let us consider the subgradient $\hat g_{i,t}(x_{i,t}^*)\in\partial f_{i,t}(x_{i,t}^*)$ such that $r_{i,t}\hat g_{i,t}(x_{i,t}^*)+\pi_t^*
	= 0$.
	
	We take the derivative of the local Lagrangian $L_{i,t}(x_{i,t})$:
	\begin{align}\nonumber
		\frac{\partial 	L_{i,t}(x_{i,t})}{\partial x_{i,t}}
		=    
		\begin{cases}
			r_{i,t}\partial f_{i,t}(x_{i,t})+\pi_t^*, & \text{if $x_{i,t}<x'_{i,t}$},\\
			\pi_t^*, & \text{o.w.}
		\end{cases} 
	\end{align}
	where $\partial f_{i,t}(x_{i,t})$ is the set of subgradients of $f_{i,t}$ at $x_{i,t}$.
	Since $f_{i,t}$ is convex and decreasing, we have
	\begin{equation}\nonumber
		\begin{aligned}
			g_{i,t}(x_{i,t})\leq g_{i,t}(x'_{i,t})\leq \hat g_{i,t}(x^*_{i,t})\leq0, \forall x_{i,t}\leq x'_{i,t} \leq x^*_{i,t},
		\end{aligned}
	\end{equation}
	where $g_{i,t}(x_{i,t})\in\partial f_{i,t}(x_{i,t})$ and $ g_{i,t}(x_{i,t}')\in\partial f_{i,t}(x_{i,t}')$.
	We further have 
	\begin{equation}\nonumber
		\begin{aligned}
			r_{i,t}g_{i,t}(x_{i,t})+\pi_t^*\leq r_{i,t}\hat g_{i,t}(x^*_{i,t})+\pi_t^*=0.
		\end{aligned}
	\end{equation}
	This implies that $L_{i,t}(x_{i,t})$ is decreasing when $x_{i,t}<x'_{i,t}$. 
	Furthermore, it is increasing when $x_{i,t}>x'_{i,t}$, since $\frac{\partial 	L_{i,t}(x_{i,t})}{\partial x_{i,t}}=\pi_t^*\geq 0$. 
	Therefore, $x'_{i,t}$ is a minimizer of the local Lagrangian $L_{i,t}(x_{i,t})$.
\end{proof}

\section{Proof of Theorem \ref{Thm: regret}} \label{Appendix: proof}
\begin{proof}
Let $G_t$ denote a vector whose elements are
\begin{equation}\nonumber
	\begin{aligned}
		G_{i,t}=    
		\begin{cases}
			x_{i,t}-x'_{i,t}, & \text{if $i\neq s_t$},\\
			-\sum_{j\neq s_t}(x_{j,t}-x'_{j,t}), & \text{o.w.}
		\end{cases} 
	\end{aligned}
\end{equation} 
Then, the updating rule in (\ref{eq: updateA}) -- (\ref{eq: updateB}) can be summarized as 
\begin{equation}\label{eq: rule}
	\begin{aligned}
		\mathbf{x}_{t+1}\leftarrow \mathbf{x}_{t}-\alpha  G_t.
	\end{aligned}
\end{equation}
It is easy to see that (\ref{eq:sum}) and (\ref{eq:x_i}) imply $||\mathbf{x}_t||^2\leq 1$, so we also have $||G_t||^2\leq 1$.
The intuition to bound the dynamic regret is to find a relationship between $G_t$ and gradient $g_t(\mathbf{x}_t)$.
However, we first require the following two lemmas.

\begin{lemma}\label{lemm:property}
	Any instantaneous feasible solution $\mathbf{x}_t\in\mathcal{F}$ of the online min-max problem defined in (\ref{eq: minmax}) has the following properties:
	\begin{enumerate}
		\item $x_{s_t,t}\leq x^*_{s_t,t}$; \label{lemma: 1_1} 
		\item $x'_{i,t}\leq x_{i,t},\forall i\in\mathcal{N}$;\label{lemma: 1_2}
		\item $\sum_{i\in\mathcal{N}}x'_{i,t}\leq 1$;\label{lemma: 1_3}
		\item $x'_{i,t}\leq x^*_{i,t},\forall i\in\mathcal{N}$; \label{lemma: x_prime}
		\item $\sum_{i\neq s_t}(x_{i,t}-x_{i,t}') (x_{i,t}-x_{i,t}^*)\geq -2$, \label{lemma: multiplication}
	\end{enumerate}
	where $s_t=\arg\max_{i\in\mathcal{N}}f_{i,t}(x_{i,t})$ denotes the agent who is the straggler in round $t$.
\end{lemma}

\newenvironment{solution} {\begin{proof}{\bfseries of Lemma \ref{lemm:property} }} {\end{proof}}

\begin{solution}  
	For any feasible solution, 
	\begin{align}\nonumber
		f_{s_t,t}(x_{s_t,t})=f_t(\mathbf{x}_t)\geq f_t(\mathbf{x}^*_t)\geq f_{s_t,t}(x^*_{s_t,t}).
	\end{align}
	As $f_{s_t,t}$ is a decreasing function, we have \ref{lemma: 1_1}.
	Similarly, we have
	\begin{equation}\label{eq: prime}
		\begin{aligned}
			f_t(\mathbf{x}_t)
			=\max_{i\in\mathcal{N}}f_{i,t}(x_{i,t})
			\geq f_{i,t}(x_{i,t}), ~ \forall i \in\mathcal{N},
		\end{aligned}
	\end{equation}
	\begin{equation}\label{eq: star}
		\begin{aligned}
			f_t(\mathbf{x}_t)\geq f_t(\mathbf{x}^*_t)=\max_{i\in\mathcal{N}}f_{i,t}(x^*_{i,t})\geq f_{i,t}(x^*_{i,t}),~ \forall i \in\mathcal{N}.
		\end{aligned}
	\end{equation}
	As $f_{i,t}$ is a decreasing function, $f_{i,t}(x'_{i,t})\geq f_{i,t}(\hat x)$ holds for all $\hat x$ such that $f_{i,t}(\hat x)\leq f_t(\mathbf{x}_t)$.
	With (\ref{eq: prime}) and (\ref{eq: star}), we futher have
	\begin{equation}\nonumber
		\begin{aligned}
			f_{i,t}(x'_{i,t})\geq f_{i,t}(x_{i,t})\quad \text{and} \quad	f_{i,t}(x'_{i,t})\geq f_{i,t}(x_{i,t}^*).
		\end{aligned}
	\end{equation}
	This implies \ref{lemma: 1_2} and \ref{lemma: x_prime}. 
	From \ref{lemma: 1_2}, we further infer $\sum_{i\in\mathcal{N}}x'_{i,t}\leq \sum_{i\in\mathcal{N}}x_{i,t}\leq 1$, which gives \ref{lemma: 1_3}.
	
	To prove \ref{lemma: multiplication}, we can rewrite its left-hand side as follows:
	\allowdisplaybreaks
	\begin{align}	\nonumber
		&~\quad\sum_{i\neq s_t}(x_{i,t}-x_{i,t}') (x_{i,t}-x_{i,t}^*)\\\nonumber
		&= \sum_{i\neq s_t} (x_{i,t}^2-x_{i,t}x_{i,t}'-x_{i,t}x_{i,t}^*+x_{i,t}'x_{i,t}^*)\\\nonumber
		&=\sum_{i\neq s_t}\frac{1}{2}\left[(x_{i,t}-x_{i,t}')^2+(x_{i,t}-x^*_{i,t})^2+(x_{i,t}'+x_{i,t}^*)^2\right]\\\nonumber
		&\quad -\sum_{i\neq s_t}x_{i,t}'^2-\sum_{i\neq s_t}x_{i,t}^{*2}\geq -\sum_{i\neq s_t}x_{i,t}'^2-\sum_{i\neq s_t}x_{i,t}^{*2}\geq -2.
	\end{align}
	The last step holds since $||\mathbf{x}_t||^2\leq 1$ and $x'_{i,t}\leq x_{i,t},\forall i\in\mathcal{N}$.
\end{solution}

\begin{lemma}\label{lemma:gG}
	Any instantaneous feasible solution $\mathbf{x}_t\in\mathcal{F}$ of the online min-max problem defined in (\ref{eq: minmax}) with Assumptions \ref{eq: cvx} and \ref{eq: gradient} satisfies
	\begin{equation}\label{eq: G_g}
		\begin{aligned}
			\left[\frac{	f_t(\mathbf{x}_t)-	f_t(\mathbf{x}^*_t)}{L}\right]^2\leq2+G_t^\mathsf{T}(\mathbf{x}_t-\mathbf{x}_t^*)
		\end{aligned}
	\end{equation}
\end{lemma}

\newenvironment{solution_2} {\begin{proof}{\bfseries of Lemma \ref{lemma:gG} }} {\end{proof}}

\begin{solution_2}
	From Lemma \ref{lemm:property}, we have $\sum_{i\neq s_t}(x_{i,t}-x_{i,t}') (x_{i,t}-x_{i,t}^*)\geq -2$, $- (x_{s_t,t}-x_{s_t,t}^*)\geq 0$, and $-x_{i,t}'\geq -x_{i,t}^*$.
	Thus, we have
	\allowdisplaybreaks\begin{align}\nonumber
		G_t^\mathsf{T}(\mathbf{x}_t-\mathbf{x}_t^*)&=\sum_{i\neq s_t}(x_{i,t}-x_{i,t}')(x_{i,t}-x_{i,t}^*)\\\nonumber
		&\quad - (x_{s_t,t}-x_{s_t,t}^*)\sum_{i\neq s_t}(x_{i,t}-x_{i,t}')\\\nonumber
		&\geq -2+(x_{s_t,t}^*-x_{s_t,t})\sum_{i\neq s_t}(x_{i,t}-x_{i,t}')\\\nonumber
		&\geq -2+(x_{s_t,t}^*-x_{s_t,t})\sum_{i\neq s_t}(x_{i,t}-x_{i,t}^*)\\\nonumber
		&\geq  -2+\left(x_{s_t,t}^*-x_{s_t,t}\right)^2.
	\end{align}
	The last inequality holds since $\sum_{i\neq s_t}(x_{i,t}-x'_{i, t})\geq\sum_{i\neq s_t}(x_{i,t}-x_{i,t}^*)= (1-x_{s_t,t})-(1-x^*_{s_t,t})=x^*_{s_t,t}-x_{s_t,t}$.
	
	\noindent With Assumption \ref{eq: gradient}, we have
	\begin{align}\nonumber
	L|x_{s_t,t}^*-x_{s_t,t}| 
	\geq |f_{s_t,t}(x_{s_t,t})-f_{s_t,t}(x^*_{s_t,t})|\geq |f_t(\mathbf{x}_t)-	f_t(\mathbf{x}^*_t)|.
\end{align}

	
	
	\noindent Therefore, we further have
	\begin{align}\nonumber
		2+G_t^\mathsf{T}(\mathbf{x}_t-\mathbf{x}_t^*)&\geq\left(x_{s_t,t}^*-x_{s_t,t}\right)^2\geq\left[\frac{	f_t(\mathbf{x}_t)-	f_t(\mathbf{x}^*_t)}{L}\right]^2.
	\end{align}
\end{solution_2}

We are now ready to proceed with deriving the dynamic regret bound of \texttt{DORA}.
With the updating rule of \texttt{DORA} in (\ref{eq: rule}), we have
	\allowdisplaybreaks\begin{align}\nonumber
		&~\quad||\mathbf{x}_{t+1}-\mathbf{x}_{t}^*||^2 =||\mathbf{x}_{t}-\mathbf{x}_{t}^*-\alpha G_t||^2\\
		&=||\mathbf{x}_{t}-\mathbf{x}_{t}^*||^2+\alpha^2||G_t||^2-2\alpha G_t^\mathsf{T}(\mathbf{x}_{t}-\mathbf{x}_{t}^*). \label{eq: norm}	
	\end{align}
	By summing up (\ref{eq: G_g}) over time, we have
	\allowdisplaybreaks\begin{align}\nonumber
		&\quad\sum_{t=1}^T (f_t(\mathbf{x}_t)-f_t(\mathbf{x}_t^*))^2/L^2\leq\sum_{t=1}^T G_t^\mathsf{T}(\mathbf{x}_{t}-\mathbf{x}_{t}^*)+2\\\nonumber
		&\overset{(a)}{=}\sum_{t=1}^T\frac{1}{2\alpha}(||\mathbf{x}_{t}-\mathbf{x}_{t}^*||^2-||\mathbf{x}_{t+1}-\mathbf{x}_{t}^*||^2)+\sum_{t=1}^T\frac{4+\alpha||G_t||^2}{2}\\\nonumber
		&\overset{(b)}{\leq}\frac{1}{2\alpha}\sum_{t=1}^T(||\mathbf{x}_{t}||^2-||\mathbf{x}_{t+1}||^2)+\frac{1}{2\alpha}\sum_{t=1}^T2(\mathbf{x}_{t+1}-\mathbf{x}_{t})^\mathsf{T}\mathbf{x}_{t}^*\\\nonumber
		&\quad+\frac{T(4+\alpha)}{2}\\\nonumber
		&=\frac{1}{2\alpha}(||\mathbf{x}_{1}||^2-||\mathbf{x}_{T+1}||^2)+\frac{1}{\alpha}\mathbf{x}_{T+1}^\mathsf{T}\mathbf{x}_{T}^*-\frac{1}{\alpha}\mathbf{x}_{1}^\mathsf{T}\mathbf{x}_{1}^*\\\nonumber
		&\quad+\frac{1}{\alpha}\sum_{t=2}^T(\mathbf{x}_{t-1}^*-\mathbf{x}_{t}^*)^\mathsf{T}\mathbf{x}_{t}+\frac{T(4+\alpha) }{2}\\\nonumber
		&\overset{(c)}{\leq} \frac{3}{2\alpha}+\frac{1}{\alpha}\sum_{t=2}^T(\mathbf{x}_{t-1}^*-\mathbf{x}_{t}^*)^\mathsf{T}\mathbf{x}_{t}+\frac{T(4+\alpha)}{2}\\\nonumber
		&\overset{(d)}{\leq} \frac{3}{2\alpha}+\frac{P_T}{\alpha}+\frac{T(4+\alpha )}{2},
	\end{align}
	where $(a)$ is due to (\ref{eq: norm}), $(b)$ holds since
	\allowdisplaybreaks\begin{align}\nonumber
		&~\quad||\mathbf{x}_{t}-\mathbf{x}_{t}^*||^2-||\mathbf{x}_{t+1}-\mathbf{x}_{t}^*||^2\\\nonumber
		&=||\mathbf{x}_{t}||^2-||\mathbf{x}_{t+1}||^2+2(\mathbf{x}_{t+1}-\mathbf{x}_{t})^\mathsf{T}\mathbf{x}_{t}^*,
	\end{align}
	and $||G_t||^2\leq 1$, $(c)$ holds since $||\mathbf{x}_{t}||^2\leq 1$,  and $(d)$ holds since
	\begin{equation}\nonumber
		\begin{aligned}
			\sum_{t=2}^T(\mathbf{x}_{t-1}^*-\mathbf{x}_{t}^*)^\mathsf{T}\mathbf{x}_{t}\leq \sum_{t=2}^T||\mathbf{x}_{t-1}^*-\mathbf{x}_{t}^*||\cdot||\mathbf{x}_{t}||.
		\end{aligned}
	\end{equation}
	As the inequality $\frac{1}{T}\sum_t a_t\leq \sqrt{\frac{1}{T}\sum_t a_t^2}$ holds for any $a_t$, we further have
	\allowdisplaybreaks\begin{align}\nonumber
		\textnormal{Reg}_T^d &= \sum_{t=1}^T f_t(\mathbf{x}_t)-f_t(\mathbf{x}_t^*)\leq \sqrt{T\sum_{t=1}^T (f_t(\mathbf{x}_t)-f_t(\mathbf{x}_t^*))^2}\\\nonumber
		&\leq \sqrt{TL^2(\frac{3}{2\alpha}+\frac{P_T}{\alpha}+\frac{T(4+\alpha)}{2})}.
	\end{align}
	
\end{proof}

\end{document}